\documentclass{amsart}
\usepackage{amsmath}
\allowdisplaybreaks[4]

\theoremstyle{plain}
\newtheorem{theorem}{Theorem}
\newtheorem{corollary}[theorem]{Corollary}
\newtheorem{lemma}[theorem]{Lemma}
\newtheorem{proposition}[theorem]{Proposition}
\theoremstyle{remark}
\newtheorem*{remark*}{Remark}

\newcommand\CC{{\mathbb C}}
\newcommand\RR{{\mathbb R}}
\newcommand\into{\int_\Omega}
\newcommand\inth{\int_{\Omega^*}}
\renewcommand\[{\begin{equation}}
\renewcommand\]{\end{equation}}
\renewcommand\a{\mathfrak a}
\newcommand\elb{e_{\lambda,b}}
\newcommand\emlb{e_{-\lambda,b}}
\newcommand\WW{\mathcal W}
\newcommand\spr[1]{\langle#1\rangle}
\newcommand\wt{\widetilde}
\newcommand\elpbp{e_{\lambda',b'}}
\newcommand\emlpbp{e_{-\lambda',b'}}

\newcommand\intah{\int_{\a^*}}
\renewcommand\AA{\mathcal A}
\newcommand\II{\mathcal I}

\newcommand\tD{\widetilde\Delta}
\renewcommand\d{\mathbf d}
\newcommand\BB{\mathcal B}

\begin{document}

\title[Wigner transform]{Wigner transform and pseudodifferential\\
  operators on symmetric spaces\\of non-compact type}

\author{S.~Twareque Ali}
\address{Department of Mathematics and Statistics, Concordia University,
 Montr\'eal, \newline Qu\'ebec, Canada~H4B~1R6}
\email{stali{@}mathstat.concordia.ca}

\author{Miroslav~Englis}
\address{Mathematics Institute, \v Zitn\'a~25, 11567 Prague 1,
 Czech Republic\newline {\rm and} Mathematics Institute, Na~Rybn\'\i\v cku~1,
 74601~Opava, Czech Republic}
\email{englis{@}math.cas.cz}
\thanks{The second author was supported by GA~\v CR grant no.~201/09/0473
and AV~\v CR research plan no.~AV0Z10190503.}

\maketitle
\begin{abstract}We obtain a general expression for a Wigner transform (Wigner function) on symmetric spaces of non-compact type and study the Weyl calculus of pseudodifferential operators on them.
\end{abstract}

\section{Introduction}
The~Wigner transform and the Weyl calculus of pseudodifferential operators
have long played prominent roles in PDE theory \cite{HormW}~\cite{Shub},
time-frequency analysis \cite{Groch} \cite{Wong} \cite{Foll} and mathematical
physics~\cite{Voros}. As~their definition relies on the Fourier transform,
it~is not surprising that they have been studied most extensively in the
context of the Euclidean $n$-space. The~aim of this paper is to extend
these notions to a more general context where a version of the Fourier
transform is available: namely, to~symmetric spaces of non-compact~type,
with the Fourier-Helgason transform.

There have been several efforts in this direction before in the literature.
First of~all, there is an extensive theory of Weyl calculi for which the
symmetric domains are the phase spaces; these are special cases of the
so-called ``invariant operator calculi'' developed recently by Arazy
and Upmeier~\cite{AU}. (It~should be noted that these calculi seem
not to involve any analogue of the Wigner transform.) Our~goal here is
different in that we have the symmetric domains only as the configuration
space, i.e.~the Wigner transform and symbols of the Weyl operators are
functions on the cotangent bundles of the symmetric domains (or,~more
precisely, on the products $\Omega\times\Omega^*$, where $\Omega^*$ is
the Fourier-Helgason dual of the symmetric space~$\Omega$; the~latter
product is essentially isomorphic to the cotangent bundle~$T^*\Omega$).
In~this direction, Tate~\cite{Tate} studied the situation for the simplest
complex bounded symmetric domain, the~unit disc; generalization to the unit
ball of $\RR^n$ (realized as one-sheeted hyperboloid in~$\RR^{n+1}$)
has then been carried out by Bertola and the first author~\cite{AB}.
% Various other cases have been studied in TWAREQUE YOU MENTIONED YOU MAY
% KNOW SOME REFERENCES HERE?
We~also mention that apparently yet another
kind of the Weyl calculus for the disc, for~which the symbols also live
on the tangent bundle of the disc and which ultimately leads to the
occurrence of Bessel functions, was~introduced by Terras~\cite{Terr}
and studied by Trimeche~\cite{Tri} or Peng and Zhao~\cite{PeZh}; it~seems
unclear whether this calculus is in any way related to Tate's and ours.
(We~pause to remark that the Bessel-function Weyl calculus, however,
seems to have rather complicated behaviour under holomorphic
transformations of the unit disc.)

   In the physical literature there have been several different generalizations
of the original Wigner function \cite{wig} to non-flat configuration spaces and their phase spaces. One approach towards a generalization exploits the fact that that the original Wigner function lives on a coadjoint orbit of the Weyl-Heisenberg group and can be obtained using the
square-integrability property of its representations. A general description of this method, exploiting square-integrable group representations, may be found in \cite{krasali} and earlier references cited therein. An approach that is very close to the one adopted in the present paper has been used in \cite{alopowo1,alopowo2} to obtain Wigner functions on hyperboloids and spheres. However, the  results obtained there were on a case by case basis, while we present here a general theory. Another suggestion for a generaliztion, using the
entire dual space of the Weyl-Heisenberg group has been given in \cite{ibomanmarsiven}. The virtue of our present approach lies in its generality and the fact that our construction preserves both the marginality and unitarity properties that allowed the original Wigner function to be interpreted as a pseudo-probability distribution on phase space.

The~Wigner transform is constructed in Section~3 below, after reviewing
the necessary prerequisites on symmetric spaces in Section~2.
The~non-Euclidean Weyl calculus of pseudodifferential operators is introduced
in Section~4. The invertibility of the Wigner transform and its unitarity are
discussed in Sections~5 and~6, respectively. The~final Section~7 contains
miscellaneous concluding remarks, open problems,~etc. For~the most part,
our~approach parallels fairly directly that of Tate's in~\cite{Tate};
however, Theorem~\ref{Th9} and Corollary~\ref{CoroAA} seem to be new
even for his situation of the unit disc.

\textsc{Acknowledgement.} Large part of this work was done while the second
author was visiting the first; the~hospitality of the mathematics department
of Concordia University on this occasion is gratefully acknowledged.

\section{Bounded symmetric domains}
Recall that a connected Riemannian manifold $\Omega$ of dimension~$d$ is
called a \emph{symmetric space} if for any $x\in\Omega$ there exists a
(necessarily unique) element $s_x\in G$, the~group of isometries of~$\Omega$,
which is involutive (i.e.~$s_x\circ s_x=\text{id}$) and has $x$ as an isolated
fixed-point. One~calls $s_x$ the geodesic symmetry at~$x$. The~symmetric space
is called \emph{irreducible} if~it is not isomorphic to a Cartesian product
of another two symmetric spaces. Irreducible symmetric spaces come in
three types: \emph{Euclidean} (these are just $\RR^d$ and its quotients),
of~\emph{compact type} (the~compact ones) and of \emph{non-compact type}.
Any~symmetric space of non-compact type can be realized~as
(i.e.~is~isomorphic~to) a~domain in $\RR^d$ which is circular with respect
to the origin and convex (the~so-called Harish-Chandra realization).
Throughout the rest of this paper, we~will thus assume that $\Omega$
is of the latter form, i.e.~a~symmetric space of non-compact type in its
Harish-Chandra realization.

It~turns out that the geodesic symmetries $s_x$ in fact act transitively
on~$\Omega$, i.e.~for any $y,z\in\Omega$ there exists an $x\in\Omega$ such
that $s_x y=z$; denoting by $K=\{g\in G:g(0)=0\}$ the stabilizer in $G$ of
the origin $0\in\Omega$, it~therefore follows that $\Omega$ is isomorphic
to the coset space~$G/K$. (It~is also true that elements of $K$ are orthogonal
maps on $\RR^d$ that preserve~$\Omega$, and that $K$ is a maximal compact
subgroup of~$G$.) There exists a unique (up~to constant multiples)
$G$-invariant measure on~$\Omega$ (obtained as the projection of the Haar
measure on~$G$); we~will denote it by~$d\mu(z)$. (Thus $d\mu(z)=d\mu(g(z))$
for any $g\in G$.)

For $x\in\Omega$, there exists a unique geodesic symmetry $\phi_x\in G$ which
interchanges $x$ and the origin,~i.e.
\[ \phi_x\circ\phi_x=\text{id}, \; \phi_x(0)=x, \; \phi_x(x)=0, \label{tSX} \]
and $\phi_x$ has only isolated fixed-points. In~fact, $\phi_x$ has only one
fixed point, namely the geodesic mid-point between $0$ and~$x$; we~will denote,
quite generally, the~geodesic mid-point between some given $x,y\in\Omega$
by~$m_{x,y}$ or $m_{xy}$. (Thus the fixed point of $\phi_x$ is
precisely~$m_{x,0}$, and $\phi_x=s_{m_{x,0}}$.)

Employing the standard notation, let $G=NAK$ be the Iwasawa decomposition
of~$G$, $\a$~the Lie algebra of the maximal Abelian part~$A$, $\a^*$~its dual,
$r=\dim_\RR\a$ its~dimension (known as the \emph{rank} of~$\Omega$),
$\rho=(\rho_1,\dots,\rho_r)\in\a^*$ the sum of positive roots,
$M$ and $M'$ the centralizer and the normalizer of $A$ in~$K$, respectively,
and $W=M'/M$~the Weyl group.
% There is a standard way of identifying $\a$ and, hence, $\a^*$ with $\RR^r$;
% the~half-sum $\rho$ is the given in terms of the data $r,a,b$~by
% $$ \rho_j= \frac{(j-1)a+b+1}2, \qquad j=1,\dots,r. $$
For any $\lambda\in\a^*\cong\RR^r$ and $b$ in the coset space $B:=K/M=G/MAN$,
one defines the ``plane waves'' on~$\Omega$~by
$$ \elb(x):= e^{(i\lambda+\rho)(A(x,b))}, \qquad  x\in\Omega, $$
where $A(x,b)$ is the unique element of $\a$ satisfying, if $b=kM$ and $x=gK$,
$$ k^{-1}g \in  N \exp A(x,b) \, K  $$
under the Iwasawa decomposition $G=NAK$.

The Helgason-Fourier transform of $f\in C_0^\infty(\Omega)$ is a function
on $\Omega^*:=\a^*\times B$ ($\cong\RR^r\times K/M$) given~by
$$ \tilde f(\lambda,b):= \into f(x) \emlb(x) \,d\mu(x) . $$
For any $f\in C_0^\infty(\Omega)$ we then have the Fourier inversion formula
$$ f(x)=\intah \int_B \tilde f(\lambda,b) \elb(x) \, d\rho(\lambda,b)  $$
and the Plancherel theorem
$$ \into |f(x)|^2 \,d\mu(x) = \intah \int_B |\tilde f(\lambda,b)|^2
 d\rho(\lambda,b) . $$
Here
$$ d\rho(\lambda,b) := |c(\lambda)|^{-2} \, db \, d\lambda,  $$
where $db$ is the unique $K$-invariant probability measure on $K/M$,
$d\lambda$~is a suitably normalized Lebesgue measure on $\a^*\cong\RR^r$,
and $c(\lambda)$ is a certain meromorphic function on the complexification
$\a^{*\CC}\cong\CC^r$ of~$\a^*$ (the~Harish-Chandra~$c$-function).
From the Plancherel theorem it can be deduced, in particular, that $f\mapsto
\tilde f$ extends to a Hilbert space isomorphism of $L^2(d\mu)$ into $L^2
(\Omega^*,d\rho)$ whose image consists of functions $F(\lambda,b)$ which
satisfy a certain symmetry condition (relating the values $F(\lambda,b)$ and
$F(s\lambda,b)$ for $s$ in the Weyl group; see Corollary~VI.3.9 in~\cite{H7}.)

A~(linear) differential operator $L$ on $\Omega$ is called $G$-invariant~if
$$ L(f\circ g)=(Lf)\circ g  $$
for any $f\in C^\infty(\Omega)$ and any $g\in G$. For~any such~$L$, it~is
known that the ``plane waves'' are eigenfunctions of~$L$:
$$ L \elb = \tilde L(\lambda) \elb  $$
where $\tilde L(\lambda)$ is a polynomial in $r$ variables; that~is, each such
$L$ is a Fourier multiplier with respect to the Helgason-Fourier transform.
The~correspondence $L\mapsto\tilde L$ sets up an isomorphism between the ring
of all $G$-invariant differential operators on $\Omega$ and the ring of all
polynomials on $\RR^r\cong\a^*$ invariant under the Weyl group $W$.

The~``plane waves'' $\elb$ obey the following transformation rule under
composition with elements of~$G$:
\[  \elb\circ g = \elb(g0) \; e_{\lambda,g^{-1}b}.   \label{tET}   \]
(We~will often write $g0,gz,$~etc.~instead of $g(0),g(z)$~etc.)
It~follows from here that
$$ e_{\lambda,gb}(g0) \elb(g^{-1}0) = 1   $$
and
\[   \begin{aligned}
 d\rho(\lambda,gb) &= |\elb(g^{-1}0)|^2 \, d\rho(\lambda,b), \\
 d\rho(\lambda,b) &= |e_{\lambda,gb}(g0)|^2 \,d\rho(\lambda,gb).
 \end{aligned}   \label{tRH}    \]
Indeed, from the formula for the Helgason-Fourier transform and (\ref{tET})
we have
\begin{align*}
\tilde f(\lambda,gb) &= \into f(z) \, e_{-\lambda,gb}(z) \, d\mu(z)   \\
&= \into f(z) \, \frac{\emlb(g^{-1}z)}{\emlb(g^{-1}0)} \, d\mu(z) \\
&= \frac1{\emlb(g^{-1}0)} \into f(gz) \, \emlb(z) \, d\mu(z)  \\
&= \frac{(f\circ g)^\sim (\lambda,b)}{\emlb(g^{-1}0)},  \end{align*}
whence from
\begin{align*}
f(z) &= \inth \tilde f(\lambda,b) \,\elb(z) \,d\rho(\lambda,b)\\
&= \inth \tilde f(\lambda,gb) \,e_{\lambda,gb}(z) \,d\rho(\lambda,gb) \\
&= \inth \frac{(f\circ g)^\sim(\lambda,b)}{\emlb(g^{-1}0)} \;
   \frac{\elb(g^{-1}z)}{\elb(g^{-1}0)} \, d\rho(\lambda,gb)   \end{align*}
we~get, upon replacing $f$ by $f\circ g^{-1}$ and $z$ by~$gz$,
$$ f(z) = \inth \frac{\tilde f(\lambda,b)}{\emlb(g^{-1}0)} \;
   \frac{\elb(z)}{\elb(g^{-1}0)} \, d\rho(\lambda,gb),   $$
proving the claim. (Note that $\emlb=\overline{\elb}$.)

Since $|\elb(x)|^2=e^{2\rho(A(x,b))}$ does not depend on~$\lambda$,
(\ref{tRH})~in~fact implies that
\[ \begin{aligned}
 d(gb) &= |\elb(g^{-1}0)|^2 \, db, \\
 db &= |e_{\lambda,gb}(g0)|^2 \, d(gb).
\end{aligned}  \label{tEX}   \]

A~function $f$ on $\Omega$ is called \emph{$K$-invariant} if $f(kx)=f(x)$
for all $x\in\Omega$ and $k\in K$. For~such functions, the Helgason-Fourier
transform $\tilde f(\lambda,b)$ does not depend on~$b$, and reduces to the
\emph{spherical transform}
$$ \tilde f(\lambda) = \into f(z) \, \Phi_{-\lambda}(z) \,d\mu(z),  $$
where $\Phi_\lambda$ are the \emph{spherical functions}
$$ \Phi_\lambda(z) := \int_K \elb(kz) \,dk = \int_K e_{\lambda,kb}(z) \,dk. $$
One~has $\Phi_\lambda=\Phi_{s\lambda}$ for all $s$ in the Weyl group,
i.e.~$\tilde f$ is $W$-invariant. The~Fourier inversion formula and the
Plancherel theorem assume the form
\begin{gather}
f(z) = \intah \tilde f(\lambda) \, \Phi_\lambda(z) \, d\rho(\lambda),
 \label{tSPH}  \\
\into |f(z)|^2 \, d\mu(z) = \intah |\tilde f(\lambda)|^2 \,d\rho(\lambda),
 \nonumber
\end{gather}
respectively, where (abusing notation a little)
$$ d\rho(\lambda) := |c(\lambda)|^{-2} \, d\lambda.  $$

\noindent{\bf Some examples.}
{\bf 1.} The~absolutely simplest example of the type of symmetric space studied
here could be the unit interval $\Omega=(-1,1)\subset\RR$, on~which $G=O(1,1)/
\RR$ acts~by
$$ gx = \frac{ax+b}{cx+d}, \qquad x\in\Omega, \;
 g=\begin{pmatrix} a&b\\c&d\end{pmatrix} \in O(1,1),  $$
that~is,
\[ gx = \epsilon \frac{x \cosh t+\sinh t}{x\sinh t+\cosh t}, \qquad
 x\in\Omega, \; t\in\RR, \; \epsilon\in\{\pm1\}.  \label{eEG}   \]
In~particular,
$$ \phi_a x = \frac{a-x}{1-ax}, \qquad x,a\in\Omega.   $$
The stabilizer of the origin is $K=O(1)=\{\pm1\}$, the invariant measure
is $d\mu(x)=\frac{dx}{1-x^2}$, and $N=\{1\}$, $A=G$, $M=K$, $B=\{1\}$.
The~Lie algebra $\mathfrak g$ can be identified with~$\RR$, and the exponential
map $\mathfrak g\to G$ is
\[ \xi \longmapsto \tanh \xi.   \label{eEX}  \]
It~follows that
$$ \elb(x) = \Big( \frac{1+x}{1-x} \Big) ^{i\lambda/2}, \qquad\lambda\in\RR. $$
The invariant differential operators on $\Omega$ are precisely the polynomials
in $\tD:=((1-x^2)\frac{\partial}{\partial x})^2$, and
$$ \tD \elb = -\lambda^2 \elb.  $$
However, this example is not really a symmetric space of noncompact type,
since, by dimensional reasons, the Lie algebra $\mathfrak g$ is necessarily
abelian and thus $\Omega$ is actually a Euclidean space. In~fact,
the~exponential map (\ref{eEX}) gives an isomorphism of $\RR$ onto $\Omega$
under which the action (\ref{eEG}) becomes just the Euclidean motion
$\xi\mapsto\epsilon(\xi+t)$, $d\mu(x)$~becomes the Lebesgue measure~$d\xi$,
$\tD$~becomes $\partial^2/\partial\xi^2$, and $\elb(x)$ reduces to the ordinary
exponential $e^{i\lambda\xi}$. Since the Weyl group is just $W=\{\pm1\}$ while
$\rho=0$ and $c(\lambda)\equiv1$, the Helgason-Fourier transform on $\Omega$
thus reduces just to the ordinary Fourier transform on~$\RR$.

{\bf 2.} The~simplest genuine example is thus the unit disc $\Omega=\{z\in\CC
\cong\RR^2: |z|<1\}$, considered by Tate~\cite{Tate}. In~this case $\Omega=
G/K$ with $G=U(1,1)/\CC$ acting again~by
$$ gz = \frac{az+b}{cz+d}, \qquad z\in\Omega, \;
 g=\begin{pmatrix} a&b \\c&d \end{pmatrix} \in U(1,1),   $$
and $K=U(1)$,
$A=\{\begin{pmatrix} \cosh t&\sinh t\\ \sinh t & \cosh t\end{pmatrix}:
t\in\RR\}$ is~the same as in the preceding example, $M=\{1\}$, $W=\{\pm1\}$
and $\rho=1$. The~geodesic symmetries are given~by
$$ \phi_a z = \frac{a-z}{1-\overline a z}.   $$
The~quotient space $B=K/M$ can be identified with the unit circle~$\mathbf T$,
and
$$ \elb(z) = \Big( \frac{1-|z|^2}{|z-b|^2} \Big) ^{\frac{1+i\lambda}2},
 \qquad \lambda\in\RR, \; b\in\mathbf T, \; z\in\Omega.  $$
The invariant measure is $d\mu(z)=(1-|z|^2)^{-2}\,dz\wedge d\overline z$,
the invariant differential operators are precisely the polynomials in
$\tD:=(1-|z|^2)^2\Delta$, where $\Delta$ is the ordinary Laplace operator, and
$$ \tD \elb = -(\lambda^2+1) \elb .   $$
The~Plancherel measure $d\rho$ is given by $d\rho(\lambda)=\frac\lambda{4\pi}
\tanh\frac{\pi\lambda}2\,d\lambda$, yielding the simplest nontrivial example
of the Helgason-Fourier transform.

{\bf 3.} The real hyperbolic $n$-space, modelled in \cite{AB} as one-sheeted
hyperboloid, can~be realized as the unit ball $\Omega=\{x\in\RR^n:|x|<1\}=G/K$
with $G=O(n,1)/\RR$, $K=O(n)$. The geodesic symmetries are the Moebius maps
$$ \phi_a x = \frac{(1-2\spr{a,x}+|x|^2)a-(1-|a|^2)x}{1-2\spr{a,x}+|a|^2|x|^2},
 \qquad x,a\in\Omega ;  $$
the maximal abelian subgroup $A$ can be identified with $\{\tau_a: a=re_1,\;
-1<r<1\}$ where $\tau_a(x):=\phi_a(-x)$ and $e_1=(1,0,0,\dots,0)$; and
$M=\{k\in K: ke_1=e_1\}\cong O(n-1)$, so~that $B=K/M$ can again be identified
with the unit sphere $\partial\Omega=\mathbf S^{n-1}$. The~Weyl group $W$ is
again just $\{\pm1\}$, the sum of positive roots is $\rho=n-1$, and the
``plane waves'' are
$$ \elb(x) = \Big( \frac{1-|x|^2}{|x-b|^2} \Big)^{\frac{n-1+i\lambda}2},
 \qquad x\in\Omega, \; b\in\partial\Omega, \; \lambda\in\RR.  $$
The invariant differential operators are precisely the polynomials~in
$$ \tD := (1-|x|^2) \Big[ (1-|x|^2) \sum_{j=1}^n \frac{\partial^2}
{\partial x_j^2} + (2n-4)\sum_{j=1}^n x_j \frac\partial{\partial x_j} \Big], $$
and $\tD\elb=-(\lambda^2+(n-1)^2)\elb$. Note that for $n=1$ and $n=2$, this
example recovers the previous two as special cases.

{\bf 4.} All~three examples above are in turn special cases of the unit ball
of real $n\times m$ matrices
$$ \Omega = \{ z\in\RR^{n\times m}: I-z^t z\text{ is positive definite}\}  $$
(or,~equivalently, $\|z\|<1$ when $z$ is viewed as an operator
$z:\RR^m\to\RR^n$). One~has $\Omega=G/K$ with $G=O(n,m)/\RR$ acting~by
$$ gz = (az+b)(cz+d)^{-1}, \qquad z\in\Omega, \;
 g=\begin{pmatrix} a&b \\ c&d \end{pmatrix} \in O(n,m)  $$
(with $a\in\RR^{n\times n}$, $b\in\RR^{n\times m}$, etc.). The~stabilizer
subgroup $K$ consists of all block-diagonal ($b=c=0$) elements in~$G$, while
$A$ can be taken as $\{\tau_a: a=\sum_j r_je_j, \; -1<r_j<1\}$, where
$\tau_a(x):=\phi_a(-x)$ and $e_j$ is the $n\times m$ matrix with 1 on the
$(j,j)$-position and 0 everywhere else, $1\le j\le\min(m,n)$. In~particular,
the rank of $\Omega$ is $r=\min(m,n)$. (The~previous three examples,
corresponding to~$m=1$, were thus of rank~1.)

{\bf 5.} General symmetric spaces of non-compact type include, in~addition to
analogous unit balls of symmetric or anti-symmetric matrices, also some other
infinite series of matrix domains, as~well as so-called ``exceptional''
symmetric domains related to (some) exceptional Lie groups.

\smallskip

For~more details and the proofs of all the above, as~well as for the
complete classification (up~to isomorphism) of all symmetric spaces,
we~refer e.g.~to Helgason's books~\cite{H7},~\cite{HeGA},~\cite{HeDS}.

\section{Wigner transform}
Recall that $m_{x,y}$ stands for the geodesic midpoint between two points
$x,y$ of~$\Omega$. We~begin by establishing a few properties of the Jacobian
$J(x,y)$ of this map, defined by the following equality
\[  \into f(m_{z,y}) \, d\mu(z) = \into f(x) \; J(x,y) \, d\mu(x).
  \label{tJA}  \]

\begin{proposition} For any $g\in G$, $J(gx,gy)=J(x,y)$.
\end{proposition}

\begin{proof} From the definition of $J$ and invariance of $d\mu$ we~get
\begin{align*}
\into f(x) \, J(gx,gy) \, d\mu(x)
&= \into f(x) \, J(gx,gy) \, d\mu(gx) \\
&= \into f\circ g^{-1}(x) \, J(x,gy) \, d\mu(x)  \\
&= \into f\circ g^{-1}(m_{z,gy}) \, d\mu(z)  \\
&= \into f\circ g^{-1}(gm_{g^{-1}z,y}) \, d\mu(z)  \\
&= \into f(m_{g^{-1}z,y}) \, d\mu(z)  \\
&= \into f(m_{z,y}) \, d\mu(z)  \\
&= \into f(x) \, J(x,y) \, d\mu(x) ,   \end{align*}
where the fourth equality follows from the fact that $m_{gx,gy}=gm_{x,y}$.
\end{proof}

\begin{corollary} $J(x,y)=J(y,x)$.    \end{corollary}

\begin{proof} Take for $g$ the geodesic symmetry interchanging $x$ and~$y$.
\end{proof}

Now~let $F$ be a function on $\Omega\times\Omega$. The~\emph{Wigner transform}
$\WW_F:\Omega\times\Omega^*\to\CC$ of $F$ is defined~by
\begin{align*}
\WW_F(x;\lambda,b) :&= |\elb(x)|^{-2} \into \elb(y) \; \emlb(s_x y) \;
   F(s_x y,y) \; J(x,y) \, d\mu(y)   \\
&= |\elb(x)|^{-2} \into \elb(s_x y) \; \emlb(y) \;
   F(y,s_x y) \; J(x,y) \, d\mu(y) .   \end{align*}
The~second expression follows from the first upon the change of variable
$y\mapsto s_x y$ and noting that $J(x,y)=J(s_xx,s_xy)=J(x,s_xy)$ by the
preceding proposition.

Note that the quantity $|\elb(x)|^{-2}$~is, in~fact, independent of~$\lambda$.

The~next three theorems show that our Wigner transform retains the properties
we expect from the Euclidean case.

\begin{theorem} {\rm(Invariance)} For any $g\in G$,
$$ \WW_{F\circ g}(x;\lambda,b) = \WW_F(gx;\lambda,gb),   $$
where $F\circ g(x,y):=F(gx,gy)$.   \end{theorem}

\begin{proof} Note that for any $x,y\in\Omega$ and $g\in G$,
$$ s_{gx}gy=g s_x y.  $$
Using~the definition of~$\WW$, the invariance of~$d\mu$ and~$J$, and
(\ref{tET}), we~therefore have
\begin{align*}
& \WW_{F\circ g}(g^{-1}x;\lambda,b) \\
&\qquad
 = |\elb(g^{-1}x)|^{-2} \into \elb(y) \emlb(s_{g^{-1}x}y) F(gs_{g^{-1}x}y,gy)
 J(g^{-1}x,y) \, d\mu(y)  \\
&\qquad= |\elb(g^{-1}x)|^{-2} \into \elb(g^{-1}y) \emlb(s_{g^{-1}x}g^{-1}y)
 F(gs_{g^{-1}x}g^{-1}y,y) \\
&\hskip24em
 J(g^{-1}x,g^{-1}y) \, d\mu(y)  \\
&\qquad= |\elb(g^{-1}x)|^{-2} \into \elb(g^{-1}y) \emlb(g^{-1}s_x y)
 F(s_x y,y) J(x,y) \, d\mu(y)  \\
&\qquad=
 |\elb(g^{-1}0) e_{\lambda,gb}(x)|^{-2} \into \elb(g^{-1}0) e_{\lambda,gb}(y)
 \emlb(g^{-1}0) e_{-\lambda,gb}(s_x y) \\
&\hskip24em  F(s_x y,y) J(x,y) \, d\mu(y)  \\
&\qquad=
 |e_{\lambda,gb}(x)|^{-2} \into e_{\lambda,gb}(y) e_{-\lambda,gb}(s_x y)
 F(s_x y,y) J(x,y) \, d\mu(y)  \\
&\qquad= \WW_F(x;\lambda,gb),  \end{align*}
as~asserted.   \end{proof}

\begin{theorem} {\rm(Marginality)} For~$F$ of the form $F(x,y)=f(x)\overline
{g(y)}$, with $f,g\in L^2(\Omega,d\mu)$, we~have the marginality relations
\begin{align*}
\into \WW_F(x;\lambda,b) \, |\elb(x)|^2 \, d\mu(x) &=
 \tilde f(\lambda,b) \, \overline{\tilde g(\lambda,b)} ;  \\
\inth \WW_F(x;\lambda,b) \, |\elb(x)|^2 \, d\rho(\lambda,b) &=
 f(x) \, \overline{g(x)}.   \end{align*}
\end{theorem}

\begin{proof} For~the first, use the defining property (\ref{tJA}) of the
Jacobian:
\begin{align*}
& \into \WW_F(x;\lambda,b) \, |\elb(x)|^2 \, d\mu(x)  \\
&\qquad\qquad =
\into\into \elb(y) \; \emlb(s_xy) \; f(s_xy) \; \overline{g(y)}
 \; J(x,y) \, d\mu(x) \, d\mu(y)  \\
&\qquad\qquad =
\into\into \elb(y) \; \emlb(s_{m_{z,y}}y) \; f(s_{m_{z,y}} y) \;
 \overline{g(y)} \, d\mu(z) \, d\mu(y)  \\
&\qquad\qquad =
\into\into \elb(y) \; \emlb(z) \; f(z) \; \overline{g(y)} \, d\mu(z)
 \, d\mu(y)  \qquad\text{(since } s_{m_{z,y}}y=z)   \\
&\qquad\qquad =
 \tilde f(\lambda,b) \, \overline{\tilde g(\lambda,b)}.   \end{align*}
For the second, note that by Plancherel
\[ \inth \elb(y) \, \emlb(z) \, d\rho(\lambda,b) = \delta_{yz}.  \label{tPL} \]
Thus
\begin{align*}
& \inth \WW_F(x;\lambda,b) \, |\elb(x)|^2 \, d\rho(\lambda,b) \\
&\qquad\qquad =
\into\inth \elb(y) \; \emlb(s_x y) \; f(s_x y) \; \overline{g(y)}
 \; J(x,y) \, d\rho(\lambda,b) \, d\mu(y)  \\
&\qquad\qquad =
\into \delta_{y,s_xy} \; f(s_x y) \; \overline{g(y)} \; J(x,y) \, d\mu(y)  \\
&\qquad\qquad = f(s_x x) \; \overline{g(x)} \; J(x,x)   \\
&\qquad\qquad = f(x) \; \overline{g(x)} \; J(x,x).   \end{align*}
On~the other hand, by the invariance of~$J$ we have $J(x,x)=J(\phi_xx,\phi_xx)
=J(0,0)$, and taking in the defining property for~$J$
$$ \into f(x) \; J(x,0) \, d\mu(x) = \into f(m_{z,0}) \, d\mu(z)  $$
for $f$ an approximate identity (i.e.~letting $f$ tend to the delta function
at the origin), we~get $J(0,0)=1$. Thus the second part of the theorem follows.
\end{proof}

\begin{remark*}
In~addition to the Iwasawa decomposition $G=NAK$, one~also has the
Bruhat decomposition $G=\overline{KA^+K}$, where $A^+$ is a certain
``positive'' subset of~$A$ and the bar stands for closure. It~can be deduced
from the latter that the ambient space $\RR^d\supset\Omega=G/K$ admits a
``polar decomposition'' as $\RR^d\cong\overline{K/M\times\a^+}$ ---
more precisely, any $x\in\RR^d$ can be written in the form $x=ka$ with
$a$ lying in a fixed subspace isomorphic to~$\a\cong\a^*\cong\RR^r$;
and if we set $\a^+=\{t_1 e_1+\dots+t_r e_r: \; t_1>t_2>\dots>t_r>0\}$,
where $e_1,\dots,e_r$ is an appropriate basis for~$\a$, then the correspondence
$\RR^d\ni x\longleftrightarrow(kM,a)\in K/M\times\overline{\a^+}$ is one-to-one
except for the set of measure zero where $t_j=t_{j+1}$ or $t_j=0$ for some~$j$
(then the $t_1,\dots,t_r$ are still determined uniquely, but $kM$ is~not).
(The~$r$-tuple $\d(x):=(t_1,\dots,t_r)$ is called the ``complex distance''
of $x$ from the origin.)
In~this way, the cotangent space $T^*_x\Omega\cong\RR^d$
at any point $x\in\Omega$ can essentially be identified with $K/M\times\a^*$,
and we~can thus think of the Fourier-Helgason transform
$\tilde f:\Omega^*\to\CC$ as living on the cotangent space~$T^*_x\Omega$.
Similarly, the Wigner transform $\WW_F:\Omega\times\Omega^*\to\CC$
can be envisaged as living in fact on the cotangent bundle $T^*\Omega$.
In~a~way, this is reminiscent of viewing the ordinary Fourier transform
$\tilde f(\xi)$ on $\RR^2\cong\CC$ in the polar coordinates as $\tilde
f(\xi)\equiv\tilde f(r,\theta)$ where $\xi=re^{i\theta}$; the~subtle
difference is that instead of the simple symmetry relation $\tilde f
(r,\theta)=\tilde f(-r,\theta+\pi)$, for~the Fourier-Helgason transform
one has the more complicated symmetry relations, mentioned in Section~2,
relating $\tilde f(\lambda,b)$ and $\tilde f(s\lambda,b)$ for $s$ in the
Weyl group.  \qed   \end{remark*}

\section{Pseudodifferential operators}
In~analogy with the Euclidean case, the Wigner function can be used to
define the Weyl calculus of pseudodifferential operators by assigning
to a ``symbol'' function $a$ on $\Omega\times\Omega^*$ the operator
$\Psi_a$ on $L^2(\Omega,d\mu)$ defined~by
$$ \spr{\Psi_a u,v} = \into\inth \WW_{u\otimes\overline v}(x;\lambda,b) \;
 a(x,\lambda,b) \; |\elb(x)|^2 \, d\rho(\lambda,b) \, d\mu(x),   $$
where $(u\otimes\overline v)(x,y):=u(x)\overline{v(y)}$. In~other words,
\begin{align*}
\Psi_a u(y) &= \inth \into a(m_{z,y};\lambda,b) \; \elb(y) \; \emlb(z)
 \; u(z) \, d\mu(z) \, d\rho(\lambda,b)   \\
&= \into \wt a(y,z) \; u(z) \, d\mu(z),    \end{align*}
where $\wt a$ is the integral kernel
\[  \wt a(y,z) := \inth a(m_{z,y};\lambda,b) \; \elb(y) \; \emlb(z)
 \, d\rho(\lambda,b).    \label{tIK}   \]
Note that for $a(x;\lambda,b)=a(x)$ depending only on the space variable,
$\Psi_a$~reduces just to a multiplication operator: indeed, by~Plancherel's
formula~(\ref{tPL}),
$$ \Psi_a u(y) = \into a(m_{z,y}) \; \delta_{y,z} \; u(z) \, d\mu(z)
 = a(m_{y,y}) u(y) = a(y) u(y).   $$
Similarly, for $a(x;\lambda,b)=a(\lambda)$ depending only on $\lambda$
the operator $\Psi_a$ reduces to the corresponding Fourier multiplier:
$$ \Psi_a u(y) = \inth a(\lambda) \; \elb(y) \; \tilde u(\lambda,b)
 \, d\rho(\lambda,b) = \Big(a(\lambda)\tilde u(\lambda,b)\Big){}^\wedge,  $$
where ${}^\wedge$ stands for the inverse Fourier-Helgason transform.
This shows, in~particular, that all invariant differential operators on
$\Omega$ arise as $\Psi_a$ for $a=a(\lambda)$ an appropriate $W$-invariant
polynomial on~$\a^*$.

The~invariance properties of the Wigner transform are reflected in the
corresponding invariance properties for the Weyl pseudodifferential operators
$\Psi_a$ and their integral kernels~$\wt a$.

\begin{theorem} For any $g\in G$, we~have
$$ \wt a(gy,gz) = \wt{a^g}(y,z),   $$
i.e.~$\wt a\circ g=\wt{a^g}$, where $a^g(x;\lambda,b):=a(gx;\lambda,gb)$.
\end{theorem}

\begin{proof} From~(\ref{tET}),
\begin{align*}
\wt a(gy,gz) &= \inth a(gm_{z,y};\lambda,b) \; \elb(gy) \; \emlb(gz)
 \, d\rho(\lambda,b)  \\
&= \inth |\elb(g0)|^2 a(gm_{z,y};\lambda,b) \; e_{\lambda,g^{-1}b}(y) \;
 e_{-\lambda,g^{-1}b}(z) \, d\rho(\lambda,b)  \\
&= \inth a(gm_{z,y};\lambda,gb) \; \elb(y) \; \emlb(z) \;|e_{\lambda,gb}(g0)|^2
 \, d\rho(\lambda,gb)  \\
&= \inth a(gm_{z,y};\lambda,gb) \; \elb(y) \; \emlb(z) \, d\rho(\lambda,b)  \\
&= \wt{a^g}(y,z),    \end{align*}
where the penultimate equality used~(\ref{tRH}).    \end{proof}

For $g\in G$, let $U_g$ denote the unitary operator on $L^2(\Omega,d\mu)$
of composition with~$g^{-1}$:
$$ U_g f(z) := f(g^{-1}z).  $$

\begin{theorem} $U_g^*\Psi_a U_g=\Psi_{a^g}$.   \end{theorem}

\begin{proof} Using the invariance of~$d\mu$ and the preceding theorem, we~get
\begin{align*}
\spr{\Psi_a U_g u,U_g v}
&= \into \into
 \wt a(y,z) \; u(g^{-1}y) \;\overline{v(g^{-1}z)} \,d\mu(y)\,d\mu(z) \\
&= \into \into \wt a(gy,gz) \; u(y) \;\overline{v(z)} \,d\mu(y) \,d\mu(z) \\
&= \into \into \wt{a^g}(y,z) \; u(y) \;\overline{v(z)} \,d\mu(y) \,d\mu(z) \\
&= \spr{\Psi_{a^g} u,v},   \end{align*}
completing the proof.   \end{proof}

\section{Invertibility}
We~proceed by showing that, to~a~certain extent, the assignments $a\mapsto
\wt a$ and $F\mapsto\WW_F$ are inverses of each other. While in the Euclidean
case this is true without any restrictions, for~symmetric domains this turns
out to hold, in~general, only for functions $F$ which are of a special form.
On~the level of the Wigner transform, this corresponds to symbols $a$ on
$\Omega\times\Omega^*$ which are independent of the variable~$b$:
$$ a(x;\lambda,b) = a(x;\lambda).   $$

\begin{theorem} \label{Th7}
Let $a$ be a function on $\Omega\times\Omega^*$ which is independent
of the variable~$b$. Then
$$ \WW_{\wt a} = a.   $$
\end{theorem}

We~begin with a lemma.

\begin{lemma} \label{Le8} For $x,y\in\Omega$ and $\lambda\in\a^*$,
$$ \int_B \elb(x) \; \emlb(y) \, db = \Phi_\lambda(\phi_y x)
 =\Phi_\lambda(\phi_x y).  $$
\end{lemma}

\begin{proof}
By~(\ref{tET}),
$$ \elb(x) = \elb(\phi_y \phi_y x) = \elb(\phi_y0) e_{\lambda,\phi_yb}
 (\phi_y x) = \elb(y) e_{\lambda,\phi_yb} (\phi_y x).  $$
Hence
$$ \elb(x) \emlb(y) = |\elb(y)|^2 \; e_{\lambda,\phi_yb} (\phi_y x).  $$
But~by~(\ref{tEX}), $|\elb(y)|^2\,db=d(\phi_yb)$; thus
\begin{align*}
\int_B \elb(x) \; \emlb(y) \, db
&= \int_B e_{\lambda,\phi_yb} (\phi_y x) \, d(\phi_y b)  \\
&= \int_B \elb(\phi_y x) \, db = \Phi_\lambda(\phi_y x),
\end{align*}
proving the first claim. For~the second, note that $\phi_{\phi_x y}\phi_x
\phi_y=:k$ maps $0$ to~$0$, hence belongs to~$K$; and from $\phi_{\phi_x y}
=k\phi_y\phi_x$ we then get $\phi_x y=\phi_{\phi_x y}0=k\phi_y\phi_x0
=k\phi_y x$. Since $\Phi_\lambda$ is $K$-invariant, it~follows that
$\Phi_\lambda(\phi_xy)=\Phi_\lambda(\phi_yx)$.     \end{proof}

\begin{proof}[Proof of Theorem~\ref{Th7}] Note that from the transformation
properties of $\WW$ and~$\wt a$ we~have
\begin{align*}
\WW_{\wt a}(x;\lambda,b) &= \WW_{\wt a}(\phi_x0;\lambda,\phi_x\phi_xb)
= \WW_{\wt a\circ\phi_x}(0;\lambda,\phi_x b)
= \WW_{\wt{a^{\phi_x}}} (0;\lambda,\phi_x b),   \\
a(x;\lambda,b) &= a(\phi_x0;\lambda,\phi_x\phi_xb)
= a^{\phi_x}(0;\lambda,\phi_x b).   \end{align*}
Furthermore, it~is immediate from the definition that if $a$ is independent
of~$b$, then so is $a^g$ for any~$g\in G$.
Thus it is enough to prove the assertion for $x=0$, i.e.~to prove that
$$ \WW_{\wt a}(0;\lambda,b) = a(0;\lambda) \qquad\forall b\in B.   $$
From the definitions we get
\begin{align*}
\WW_{\wt a}(x;\lambda,b)
&= |\elb(x)|^{-2} \into \elb(y) \; \emlb(s_x y) \; \wt a(s_x y,y)
 \; J(x,y) \, d\mu(y)  \\
&= |\elb(x)|^{-2} \into \inth \elb(y) \; \emlb(s_x y)
 \; a(m_{y,s_xy};\lambda',b') \elpbp(s_xy) \\
& \hskip14em
 \emlpbp(y) \; J(x,y) \, d\rho(\lambda',b') \, d\mu(y)  \\
&= |\elb(x)|^{-2} \inth a(x;\lambda',b') \into \elb(y) \; \emlb(s_x y)
 \; \elpbp(s_xy) \\
& \hskip14em
 \emlpbp(y) \; J(x,y) \, d\mu(y) \, d\rho(\lambda',b') ,
\end{align*}
since $m_{y,s_xy}=x$. Thus, as $\elb(0)=1$ for any $\lambda$ and~$b$,
\begin{align*}
 \WW_{\wt a}(0;\lambda,b) &= \inth a(0;\lambda',b') \into \elb(y) \;
 \emlb(s_0y) \; \elpbp(s_0y) \\
& \hskip14em
 \emlpbp(y) \; J(0,y) \, d\mu(y) \, d\rho(\lambda',b') .
\end{align*}
(Here, of~course, $s_0y=-y$, but we keep $s_0$ in order to avoid some extra
parenthesis below.)
As~$a(0;\lambda',b')$ is independent of~$b'$ by hypothesis, we~can carry out
the $b'$ integration, the result being by the last lemma
\begin{align*}
\WW_{\wt a}(0;\lambda,b)
&= \intah a(0;\lambda') \into \elb(y) \; \emlb(s_0y) \;
 \Phi_\lambda(\phi_ys_0y) \; J(0,y) \, d\mu(y) \, d\rho(\lambda')  \\
&= \into \elb(y) \; \emlb(s_0y) \; \check a(0;\phi_ys_0y) \; J(0,y) \,d\mu(y),
\end{align*}
where $\check a$ stands for the inverse Helgason-Fourier (or,~in~this case,
spherical) transform of $a(x;\lambda)$ with respect to~$\lambda$. Applying
the definition of the Jacobian, this becomes
$$ \WW_{\wt a}(0;\lambda,b) = \into \elb(m_{0y}) \; \emlb(s_0m_{0y}) \;
 \check a(0;\phi_{m_{0y}}s_0m_{0y}) \, d\mu(y).   $$
However, $\phi_{m_{y,0}}s_0m_{y,0}=y$, so~the last expression equals
$$ \WW_{\wt a}(0;\lambda,b) = \into \elb(m_{0y}) \; \emlb(s_0m_{0y}) \;
 \check a(0;y) \, d\mu(y).   $$
Since $\check a(0,\,\cdot\,)$ is a $K$-invariant function, we~can replace
$y$ by~$ky$, and then also integrate over~$k$. Since $m_{0,ky}=m_{k0,ky}
=km_{0,y}$ and $\elb(kz)=e_{\lambda,k^{-1}b}(z)$, this gives, using again
the last lemma,
\begin{align*}
\WW_{\wt a}(0;\lambda,b)
&= \into \int_K \elb(m_{0,ky}) \; \emlb(s_0m_{0,ky}) \, dk \; \check a(0;y)
 \, d\mu(y)  \\
&= \into \int_K e_{\lambda,k^{-1}b}(m_{0y}) \; e_{-\lambda,k^{-1}b}(s_0m_{0y})
 \, dk \; \check a(0;y) \, d\mu(y)  \\
&= \into \int_B e_{\lambda,b}(m_{0y}) \; e_{-\lambda,b}(s_0m_{0y}) \, db
 \; \check a(0;y) \, d\mu(y)  \\
&= \into \Phi_\lambda(\phi_{m_{0y}}s_0m_{0y}) \; \check a(0;y) \, d\mu(y)  \\
&= \into \Phi_\lambda(y) \; \check a(0;y) \, d\mu(y)  \\
&= a(0;\lambda),
\end{align*}
by~(\ref{tSPH}). (Note that, $K$~being a compact group, $d(k^{-1})=dk$.)
This completes the proof.   \end{proof}

\begin{remark*}
From the proof it is evident that the theorem in general cannot be expected
to hold if the $K$-invariance hypothesis is dropped.
\end{remark*}

To~state the analogue of the last theorem in the other direction, we~first
need to identify the functions $\wt a(x,y)$ corresponding to symbols
$a(x;\lambda,b)$ which are independent of~$b$.

Let $\AA$ denote the set of all functions $F$ on $\Omega\times\Omega$
of the form
\[ F(x,y) = A(m_{xy},\phi_x y),    \label{tAF}  \]
where $A:\Omega\times\Omega\to\CC$ is $K$-invariant in the second argument,
i.e.~$A(u,v)=A(u,kv)$ $\forall k\in K$.

\begin{remark*} The~map
$$ (x,y) \longmapsto (m_{xy},\phi_{m_{xy}}x)   $$
of $\Omega\times\Omega$ into itself is a diffeomorphism onto; its inverse
is given~by
$$ (m,u) \longmapsto (\phi_m u,\phi_m s_0 u).   $$
Thus every function $F$ on $\Omega\times\Omega$ can be written uniquely
in the form $F(x,y)=G(m_{xy},\phi_{m_{xy}}x)$ for some function $G$ on
$\Omega\times\Omega$. Functions in $\AA$ correspond to the $G$ which
are $K$-invariant in the second argument.

(Indeed, recalling the notion of the complex distance $\d(x)$ from the origin
mentioned in the end of Section~3, one can define also the complex distance
$\d(x,y)$ of two points $x,y\in\Omega$ by $\d(x,y):=\d(\phi_xy)=\d(\phi_yx)$.
It~is then known that $\d(gx,gy)=\d(x,y)$ for any $g\in G$ (and, conversely,
if~$\d(x,y)=\d(x_1,y_1)$, then there is $g\in G$ with $gx=x_1$, $gy=y_1$).
The~condition that a function $f(x)$, $x\in\Omega$, is~$K$-invariant means
precisely that it depends only on~$\d(x)$. Furthermore, $\d(x,s_0x)=\frac
{2\d(x)}{1+\d(x)^2}$ (where $\frac{2\d}{1+\d^2}:=(\frac{2d_1}{1+d_1^2},\dots,
\frac{2d_r}{1+d_r^2})$ if $\d=(d_1,\dots,d_r)$), and similarly $\d(x,y)=\frac
{2\d(m_{xy},x)}{1+\d(m_{xy},x)^2}$; that~is, $\d(x,y)$ and $\d(m_{xy},x)$ are
uniquely determined by each other, and similarly for $\d(x)$ and $\d(x,s_0x)$.
Hence, if $F(x,y)=A(m_{xy},\phi_xy)$ where $A$ is $K$-invariant in the second
argument, and $G(m,u):=F(\phi_m u,\phi_m s_0u)$, then $G(m,u)$ depends only
on $m_{\phi_m u,\phi_m s_0u}=m$ and $\d(\phi_m u,\phi_m s_0u)=\d(u,s_0u)=
\frac{2\d(u)}{1+\d(u)^2}$, hence only on $m$ and~$\d(u)$, so~it is
$K$-invariant in~$u$. Conversely, if $F(x,y)=G(m_{xy},\phi_{m_{xy}}x)$ where
$G$ is $K$-invariant in the second argument, then $F(x,y)$ depends only on
$m_{xy}$ and $\d(m_{xy},x)$, hence only on $m_{xy}$ and $\d(x,y)$, so~it has
the form $F(x,y)=A(m_{xy},\phi_xy)$ where $A(m,u)$ is $K$-invariant in~$u$,
i.e.~$F\in\AA$.)   \end{remark*}

\begin{proposition} \label{PrAA}
If~$a(x;\lambda,b)=a(x;\lambda)$ does not depend on~$b$, then $\wt a\in\AA$.
Conversely, every function $F$ in $\AA$ arises as $\wt a$ for a unique $a$
as~above.
\end{proposition}

\begin{proof} For $a=a(x;\lambda)$ independent of~$b$, we~have by
Lemma~\ref{Le8}
\begin{align*}
\wt a(x,y) &= \inth a(m_{xy};\lambda) \; \elb(x) \; \emlb(y)
 \, d\rho(\lambda,b)  \\
&= \inth a(m_{xy};\lambda) \; \Phi_\lambda(\phi_xy) \, d\rho(\lambda)  \\
&= \check a(m_{xy};\phi_x y),
\end{align*}
where $\check a$ has the same meaning as in the proof of Theorem~\ref{Th7}.
Thus $F=\wt a$ is of the form (\ref{tAF}) with $A=\check a$, proving the
first claim. The~inversion formula for the spherical transform gives the
second part.   \end{proof}

\begin{corollary} \label{CoroAA}
Let $F\in\AA$. Then
$$ \wt{\WW_F} = F.   $$
\end{corollary}

\begin{proof} With the $a$ from the last proposition, we~have by
Theorem~\ref{Th7}
$$ \wt{\WW_F} = \wt{\WW_{\wt a}} = \wt a = F.  $$    \end{proof}

Observe that in the proof of Theorem~\ref{Th7}, when computing $\WW_{\wt a}
(0;\lambda,b)$ we in fact never used the full hypothesis that $a(x;\lambda,b)$
is independent of~$b$, but only that $a(0;\lambda,b)$ is independent of~$b$.
We~conclude this section by recording a small corollary to this observation.

\begin{proposition}
Assume that $a$ is $K$-invariant, in~the sense that $a=a^k$ $\forall k\in K$.
Then $a(0;\lambda,b)=a(0;\lambda)$ is independent of~$b$, and
$$ \WW_{\wt a}(0;\lambda,b) = a(0;\lambda) \qquad \forall b\in B.  $$
\end{proposition}

\begin{proof} From $a^k=a$ we get
$$ a(0;\lambda,kb) = a(k0;\lambda,kb) = a(0;\lambda,b),  $$
proving that $a(0;\lambda,b)$ is independent of~$b$, since $K$ acts
transitively on $B=K/M$. The rest is immediate from the observation
preceding the proposition.    \end{proof}

\section{Unitarity}
The~classical Euclidean Wigner transform is a unitary operator on
$L^2(\RR^n\times\RR^n)$.
Here is an analogue of this fact in our setting of symmetric spaces
of non-compact type. This theorem seems not to have been hitherto
noticed even in the simplest non-Euclidean setting of the unit disc.
For~brevity, let~us denote by $\BB$ the set of all functions $a(x;\lambda,b)$
on $\Omega\times\Omega^*$ which are independent of the variable~$b$.

\begin{theorem} \label{Th9} The~map $F\mapsto\WW_F$ is a unitary operator
from $L^2(\Omega\times\Omega,d\mu\times d\mu)\cap\AA$ onto
$L^2(\Omega\times\Omega^*,d\mu\times d\rho)\cap\BB$.
\end{theorem}

\begin{proof} In~view of Theorem~\ref{Th7} and Proposition~\ref{PrAA},
it~is enough to show that the inverse map $a\mapsto\wt a$ is an isometry
from $L^2(\Omega\times\Omega^*)\cap\BB$ into $L^2(\Omega\times\Omega)\cap\AA$
(with their respective measures); that~is, that
$$ \into \intah |a(x,\lambda)|^2 \, d\rho(\lambda) \, d\mu(x) =
 \into\into |\wt a(x,y)|^2 \, d\mu(y) \, d\mu(x).   $$
In~course of the proof of Proposition~\ref{Th7}, we~have seen that
$\wt a(x,y)=\check a(m_{xy},\phi_xy)$, where $\check a$ has again
the same meaning as before. By~Plancherel, the desired equality is
therefore equivalent~to
\[ \into\into |\check a(x;y)|^2 \, d\mu(y) \, d\mu(x) =
 \into \into |\check a(m_{xy},\phi_x y)|^2 \, d\mu(y) \, d\mu(x).
 \label{tUA}   \]
In~view of the invariance of the measure~$d\mu$, we~may replace $x$ by~$kx$,
$k\in K$, so~the left-hand side of (\ref{tUA}) equals
\[ \into \into |\check a(kx,y)|^2 \, d\mu(x) \, d\mu(y), \label{tUB}  \]
for any $k\in K$. Similarly, on~the right-hand side of (\ref{tUA}) we~can
replace $y$ by $\phi_x y$ in the inner integral, giving
$$ \into \into |\check a(m_{x,\phi_x y},y)|^2 \, d\mu(y) \, d\mu(x).  $$
Replacing again $x,y$ by~$kx,ky$, $k\in K$, recalling that $\check a$ is
a $K$-invariant function in its second argument, and using the fact that
$\phi_{kx}ky=k\phi_x y$ and $m_{kx,k\phi_xy} = m_{k\phi_x0,k\phi_xy}=k\phi_x
m_{0y}$ (since $\phi_x$ is a Riemannian isometry), we~thus see that the
right-hand side of (\ref{tUA}) is~equal~to
\[ \into \into |\check a(k\phi_x m_{0y},y)|^2 \, d\mu(x) \, d\mu(y),
 \label{tUC}   \]
for any $k\in K$. Since the $k\in K$ in both (\ref{tUB}) and~(\ref{tUC})
can be taken arbitrary, the desired equality (\ref{tUA}) is therefore
actually equivalent~to (as~$x=\phi_x0$)
$$ \int_K \into \into |\check a(k\phi_x0,y)|^2 \, d\mu(x) \, d\mu(y) \, dk =
\int_K \into \into |\check a(k\phi_xm_{0y},y)|^2 \, d\mu(x) \, d\mu(y) \,dk. $$
We~claim that we in fact have even the equality
\[ \postdisplaypenalty1000000
\int_K \into |\check a(k\phi_x0,y)|^2 \, d\mu(x) \, dk =
\int_K \into |\check a(k\phi_xm_{0y},y)|^2 \, d\mu(x) \,dk  \label{tUD}  \]
for any fixed $y\in\Omega$.

Indeed, denote, for brevity, $F(x):=|\check a(x;y)|^2$. For~$z\in\Omega$,
consider the integral
$$ \II(z) := \int_K \into F(k\phi_xz) \, dk \, d\mu(x).    $$
Now any $g\in G$ can be uniquely written in the form $k\phi_x$ with $k\in K$
and $x\in\Omega$ (in~fact, $x=g^{-1}0$ and $k=g\phi_x$), and the measure
$dk\,d\mu(x)$ corresponds under this parameterization to the Haar measure
$dg$ on~$G$. (Recall that~$G$, as~a~semisimple Lie group, is~unimodular,
so~$dg$ is both the left and the right Haar measure.) Thus
$$ \II(z) = \int_G F(gz) \, dg.   $$
For any $g_1\in G$, the invariance of the Haar measure gives
$$ \II(g_1 z) = \int_G F(gg_1z) \,dg = \int_G F(gg_1z) \,d(gg_1) = \II(z).  $$
Since $G$ acts transitively on $\Omega=G/K$, $\II(z)$~is~thus independent
of~$z$. In~particular, $\II(0)=\II(m_{0y})$, proving~(\ref{tUD}) and completing
the proof of the theorem.    \end{proof}

\begin{remark*} The equality (\ref{tUA}) would follow immediately if the
diffeomorphism $(x,y)\mapsto(m_{xy},\phi_xy)$ were measure-preserving.
However, a~simple calculation shows that on the disc
$$ \frac{d\mu(m_{xy}) \, d\mu(\phi_xy)} {d\mu(x) \, d\mu(y)} =
\frac{2-\overline x y-\overline y x}{2|1-\overline x y|}
\sqrt{ \frac{1-|x|^2}{1-|y|^2} } \neq 1,   $$
so~this is not the case even for the unit disc.     \end{remark*}

The~proof of the last theorem again indicates that $\WW$ cannot probably
be expected to act unitarily also on functions $a(x;\lambda,b)$ which are
not independent of~$b$.

\begin{remark*} Note that the class of $K$-invariant functions on~$\Omega$,
and~the corresponding class of the functions on $\Omega^*$ which are
independent of~$b$, play a distinguished role also in the properties of
the Helgason-Fourier transform: for~instance, the convolution $f*g$ of two
functions on $\Omega$ does not in general satisfy $(f*g)^\sim=\tilde f\tilde g$
(which is notorious for the ordinary Fourier transform), however this becomes
true if $g$ is $K$-invariant. This makes the introduction of the two function
classes $\AA$, $\BB$ above quite natural.     \end{remark*}

\section{Concluding remarks}
\subsection{}
We~have been somewhat nonspecific about what kind of functions we are dealing
with: for~instance, the inversion formula for the Helgason-Fourier transform
holds for $f$ smooth with compact support, and extends to $f\in L^2$ only by
Plancherel. There are also analogues of the Schwartz space, one~on $\Omega$
and another one on~$\Omega^*$, such that the Helgason-Fourier transform is
an isomorphism of the former onto the latter; see e.g.~\cite{GaVa}, Chapter~6.
The~rigorously minded reader should think of the functions $a,\wt a$, etc.,
as~belonging to the appropriate tensor products of these Schwartz spaces;
in~that case the convergence of all the integrals involved etc. can be
verified with ease.
Extensions to more general functions, or~even distributions, can be achieved
by the standard techniques used for handling oscillatory integrals
(see e.g.~\cite{Shub}).

\subsection{}
One~might try introducing H\"ormander classes for symbols $a$, and building
an analogue of the usual calculus for the Weyl operators~$\Psi_a$ ---
composition formulas, boundedness in Sobolev spaces,~etc. Some steps in
this direction have been done in Tate~\cite{Tate} for the disc.

A~related theory of pseudodifferential operators (on~the disc, but very likely
extending to any symmetric space of non-compact type), corresponding to the
standard Kohn-Nirenberg, rather than Weyl, pseudodifferential operators in
the Euclidean case, was developed by Zelditch~\cite{Zeld}.
However, expressing our operators $\Psi_a$ as these ``Kohn-Nirenberg''
pseudodifferential operators (thus reducing the questions mentioned in the
previous paragraph to the theory already developed by Zelditch) does not seem
straightforward, even for the special case of functions independent of~$b$.

\subsection{}
The~function class $\AA$ is somewhat mysterious: it~is totally unclear to the
present authors, for~instance, how to characterize the Weyl operators $\Psi_a$
with $\wt a\in\AA$. In~the Euclidean setting, this would correspond to
operators whose Schwartz kernels depend only on $\frac{x+y}2$
and $|x-y|$; even in this case the answer is not obvious.

\subsection{}
Though the authors are convinced that there are no analogues of
Theorem~\ref{Th7} and Theorem~\ref{Th9} for general symbols~$a$
(i.e.~possibly depending on~$b$), we~are unable to provide an
explicit counterexample.


\begin{thebibliography}{99}

\bibitem{AB} S. Twareque Ali and M. Bertola, {\it Symplectic geometry of the Wigner transform on noncompact symmetric spaces,\/}  in {\it Group24: Physical and Mathematical Aspects of Symmetries,\/} Proceedings of the 24th International Colloquium on Group Theoretical Methods in Physics,
Paris, July 15-20, 2002, (Satellite Colloquium on Coherent States, Wavelets and Applica-
tions, Louvain la Neuve, July 10-12, 2002), Eds: J.-P. Gazeau, R. Kerner, J.-P.
Antoine, S. M\'etens and J.-Y. Thibon, Institute of Physics Conference Series Number 173,
Institute of Physics Publishing, Bristol and Philadelphia (2003), pp. 847--886.


\bibitem{alopowo1}M. A. Alonso, G. S. Pogosyan and K. B. Wolf, {\it Wigner functions for curved spaces I: On hyperboloids,\/} J. Math. Phys. {\bf 43} (2002), 5857 (15 pages); doi:10.1063/1.1518139.

\bibitem{alopowo2}M. A. Alonso, G. S. Pogosyan and K. B. Wolf, {\it Wigner functions for curved spaces II: On spheres,\/} J. Math. Phys. {\bf 44} (2003), 1472(18 pages); doi:10.1063/1.1559644.

\bibitem{AU} J. Arazy, H. Upmeier: {\it Invariant symbolic calculi and
eigenvalues of invariant operators on symmetric domains,\/} Function
spaces, interpolation theory, and related topics (Lund, 2000) (A. Kufner,
M. Cwikel, M. Engli\v s, L.-E. Persson, and G. Sparr, eds.), pp.~151--211,
Walter de Gruyter, Berlin, 2002.

\bibitem{Foll} G.B. Folland, {\it Harmonic analysis in phase space,\/}
Annals of Mathematics Studies~122, Princeton University Press, Princeton, 1989.

\bibitem{GaVa} R. Gangolli, V.S. Varadarajan: {\it Harmonic analysis of
spherical functions on real reductive groups,\/} Springer-Verlag,
Berlin-Heidelberg, 1988.

\bibitem{Groch} K. Gr\"ochenig, {\it Foundations of Time-Frequency Analysis,\/}
Birkh\"auser, Boston-Basel-Berlin, 2001.

\bibitem{HeDS} S. Helgason, {\it Differential geometry and symmetric spaces,\/}
Academic Press, New York-London, 1962.

\bibitem{HeGA} S. Helgason, {\it Groups and geometric analysis,\/}
Academic Press, Orlando, 1984.

\bibitem{H7} S. Helgason, {\it Geometric analysis on symmetric spaces,\/}
Mathematical Surveys and Monographs~39, Amer. Math. Soc., Providence, 1994.

\bibitem{HormW} L. H\"ormander, {\it The Weyl calculus of pseudodifferential
operators,\/} Comm. Pure Appl. Math. {\bf 32} (1979), 360--444.

\bibitem{ibomanmarsiven} A. Ibort, V.I. Man'ko, G. Marmo, A. Simoni and
F. Ventriglia, {\it A generalized Wigner function on the space of
irreducible representations of the Weyl-Heisenberg
group and its transformation properties,\/} J. Phys. A {\bf 42} (2009) 155302 (12pp);
doi:10.1088/1751-8113/42/15/155302.

\bibitem{krasali} A. Krasowska and S.T. Ali, {\it Wigner functions for a class of semi-direct product groups,\/} J. Phys. A {\bf 36} (2003), 2801--2820.

\bibitem{PeZh} L. Peng, J. Zhao: {\it Weyl transforms on the upper half
plane,\/} Rev. Mat. Complut. {\bf 23} (2010), 77--95.

\bibitem{Shub} M.A. Shubin, {\it Pseudodifferential operators and spectral
theory,\/} Springer-Verlag, Berlin, 1987.

\bibitem{Tate} T. Tate: {\it Weyl pseudo-differential operator and Wigner
transform on the Poincar\'e disk,\/} Ann. Global Anal. Geom. {\bf 22} (2002),
29--48.

\bibitem{Terr} A.~Terras, {\it Harmonic analysis on symmetric spaces and
applications.~I,\/} Springer-Verlag, New York, 1985.

\bibitem{Tri} L.T. Rachdi, K. Trim\'eche: {\it Weyl transforms associated with
the spherical mean operator,\/} Anal. Appl. (Singap.) {\bf 1} (2003), 141--164.

\bibitem{Voros} A. Voros: {\it An algebra of pseudodifferential operators
and the asymptotics of quantum mechanics,\/} J. Funct. Anal. {\bf 29} (1978),
104--132.

\bibitem{wig} E. Wigner, {\it On the quantum correction for thermodynamic equilibrium,\/}
Phys. Rev. {\bf 40} (1932), 749--759.

\bibitem{Wong} M.W.~Wong, {\it Weyl transforms,\/} Springer, New York, 1998.

\bibitem{Zeld} S.~Zelditch: {\it Pseudo-differential analysis on hyperbolic
surfaces,\/} J.~Funct. Anal. {\bf 68} (1986), 72--105.


\end{thebibliography}
\end{document}